\documentclass{amsproc}
\usepackage{amssymb}
\usepackage{color}
\usepackage{tikz-cd}

\copyrightinfo{2017}{The Authors}

\newtheorem{theorem}{Theorem}[section]
\newtheorem{lemma}[theorem]{Lemma} 
\newtheorem{prop}[theorem]{Proposition}

\theoremstyle{definition}

\theoremstyle{remark}

\numberwithin{equation}{section}

 \def\idty{{\mathchoice {\mathrm{1\mskip-4mu l}} {\mathrm{1\mskip-4mu l}} %
{\mathrm{1\mskip-4.5mu l}} {\mathrm{1\mskip-5mu l}}}}

\newcommand{\bE}{{\mathbb E}}
\newcommand{\bF}{{\mathbb F}}
\newcommand{\bR}{{\mathbb R}}

\newcommand{\Rl}{\mathbb{R}}
\newcommand{\Ir}{\mathbb{Z}}
\newcommand{\Cx}{\mathbb{C}}

\newcommand{\cA}{{\mathcal A}}
\newcommand{\cB}{{\mathcal B}}
\newcommand{\cH}{{\mathcal H}}
\newcommand{\cP}{{\mathcal P}}

\newcommand{\cL}{{\mathcal L}}

\newcommand{\tu}{\tilde{u}}

\newcommand{\supp}{\operatorname{supp}}
\newcommand{\tr}{\mathrm{tr}}

\newcommand{\ran}{{\rm ran}}

\newcommand{\id}{{\rm id}}

\newcommand{\infspec}{\mathrm{inf\, spec}}
\renewcommand{\span}{\mathrm{span}}

\newcommand{\be}{\begin{equation}}
\newcommand{\ee}{\end{equation}}
\newcommand{\bea}{\begin{eqnarray}}
\newcommand{\eea}{\end{eqnarray}}
\newcommand{\beann}{\begin{eqnarray*}}
\newcommand{\eeann}{\end{eqnarray*}}
\newcommand{\eq}[1]{(\ref{#1})}

\begin{document}

\title[Lieb-Robinson bounds for lattice fermion systems]{Lieb-Robinson bounds, the spectral flow, and stability of the spectral gap for lattice fermion systems}


\author[B. Nachtergaele]{Bruno Nachtergaele}
\address{Department of Mathematics and Center for Quantum Mathematics and Physics\\
University of California, Davis\\
Davis, CA 95616, USA}
\email{bxn@math.ucdavis.edu}
\thanks{Based on work supported by the National Science Foundation under Grant DMS-1515850.}

\author[R. Sims]{Robert Sims}
\address{Department of Mathematics \\
University of Arizona\\
Tuscon, AZ 85721, USA}
\email{rsims@math.arizona.edu}
\thanks{}

\author[A. Young]{Amanda Young}
\address{Department of Mathematics \\
University of Arizona\\
Tuscon, AZ 85721, USA}
\email{amyoung@math.arizona.edu}
\thanks{}


\date{April 1, 2017}

\begin{abstract}
We prove Lieb-Robinson bounds for a general class of lattice ferm\-ion systems. By making use of a suitable
conditional expectation onto subalgebras of the CAR algebra, we can apply the Lieb-Robinson bounds much in the same way as
for quantum spin systems. We preview how to obtain the spectral flow automorphisms and to prove stability of the spectral gap 
for frustration-free gapped systems satisfying a Local Topological Quantum Order condition.
\end{abstract}

\maketitle

\section{Introduction}\label{sec:introduction}

In the past dozen years, a considerable number of mathematical results on quantum spin systems 
made use of Lieb-Robinson Bounds \cite{lieb:1972} in an essential way. These include extensions
of the Lieb-Schultz-Mattis theorem to higher dimensions \cite{hastings:2004, nachtergaele:2007},
the exponential decay of correlations in gapped ground states \cite{nachtergaele:2006a,hastings:2006},
Area Laws for the entanglement entropy \cite{hastings:2007, brandao:2015}, construction of the spectral
flow and adiabatic theorems \cite{bachmann:2012,bachmann:2017c,bachmann:2017}, stability results for 
gapped ground states \cite{bravyi:2010,michalakis:2013} and more.

It is quite clear that in many cases it should be
possible to answer the same type of questions for lattice fermion systems based on the same principles
that apply to quantum spin systems. Indeed, there are several special cases where the close analogy 
between lattice systems of spins and of fermions allowed for the successful application of Lieb-Robinson
bounds. For example, a proof of the quantized Hall effect for interacting lattice fermions \cite{hastings:2015, giuliani:2017},
and linear response theory for lattice fermion systems \cite{bru:2017}.

In this paper, we discuss several crucial ingredients which enter the
proof of stability for gapped, frustration-free models of lattice fermions.
First, we prove a Lieb-Robinson bound for a general class of models.
Variants of this quasi-locality estimate enter the proof of stability in many
different stages. Next, we introduce a conditional expectation 
which enables strictly local approximations of quasi-local observables; this 
can be seen as an analogue of the normalized partial trace familiar in the 
context of quantum spin systems. Then, we provide a version of the martingale method
suitable to prove lower bounds for the spectral gap of fermion models. Equipped with this,
a wealth of potential unperturbed models can be explored. After giving a few specific
classes of gapped fermion models, we overview the main tool of analysis for stability:
the spectral flow. More detailed estimates and additional results will be provided in forth-coming papers
\cite{moon:inprep,young:inprep}.

The methods discussed in this paper can be applied to the Aubry-Andr\'e model studied by Mastropietro
in a recent series of papers \cite{mastropietro:2015,mastropietro:2016,mastropietro:2017}. Both the molecular 
and the free fermion limit can be shown to be stable under general, uniformly small, short-range perturbations. 
This is more general than the class of perturbations studied by Mastropietro, but the general result is weaker.
In particular the renormalization group method of Mastropietro takes into account the quasi-periodicity of the 
potential, in a way that allows for an estimate of the correlation length that is indicative of many-body localization.  
The details of the interactions are ignored by the general approach discussed here and adapting the method
to study many-body localization will require further research.

\section{Lattice fermion systems}\label{sec:fermions}

Spinless fermions on a countable set $\Gamma$, which is often referred to as `the lattice', 
are described by the CAR algebra $\cA_\Gamma=\mbox{CAR}(\ell^2(\Gamma))$. 
$\cA_\Gamma$ is the $C$*-algebra generated by 
creation and annihilation operators $a^*_x,a_x$, $x\in\Gamma$, which satisfy the canonical 
anti-commutation relations, i.e. the CAR:
\begin{equation} \label{CAR}
\{ a_x, a_y \} = \{a_x^*, a_y^* \} = 0 \quad \mbox{and} \quad \{a_x, a_y^* \} = \delta_{x,y} \idty \quad \mbox{for any } x,y \in \Gamma \, .
\end{equation}
Here $\{ A, B \} = AB+BA$ denotes the anti-commutator of $A$ and $B$. 
As is discussed in detail in \cite{bratteli:1997}, this CAR algebra can be represented as
the collection of bounded linear operators over the Hilbert space
corresponding to the anti-symmetric Fock space
generated by $\ell^2( \Gamma)$. 
Furthermore, we note that spin and/or band indices can be included in this description by 
extending $\Gamma$, for example, to $\tilde \Gamma = \Gamma \times
\{1,\ldots,n\}$.

For $X\subset \Gamma$, $\cA_X$ is naturally embedded as a subalgebra of $\cA_\Gamma$ by 
identifying $\ell^2(X)$ with the subspace of $\ell^2(\Gamma)$ consisting of the functions that 
vanish on $\Gamma\setminus X$. Let $\cP_0(\Gamma)$ denote the set of finite subsets of $\Gamma$.
For any $X \in \cP_0( \Gamma)$, it is useful to define the parity automorphism of $\cA_X$, which we denote by $\Theta_X$, by setting
\be \label{Parity_Op}
\Theta_X (A) = (-1)^{N_X} A (-1)^{N_X} \quad \mbox{for any } A\in \cA_X \, .
\ee
Here $N_X  = \sum_{x\in X} a^*_x a_x$ is the local number operator. 
Using the quasi-local structure of $\cA_\Gamma$, we see that 
there is a unique automorphism $\Theta$ of $\cA_\Gamma$ for which $\Theta\restriction_{\cA_X} = \Theta_X$ for any $X\in\cP_0(\Gamma)$. It is clear that $\Theta^2 = \id$.
The {\em even} and {\em odd} elements of $\cA_\Gamma$
are the eigenvectors of $\Theta$ with eigenvalue $1$ and $-1$, respectively. By $\cA_\Gamma^+$ and 
$\cA_\Gamma^-$, we denote the corresponding eigenspaces. Similarly, for any $\Lambda \in \cP_0( \Gamma)$, 
we set $\cA_{\Lambda}^+$ and $\cA_{\Lambda}^-$ to be the even and odd eigenspaces of 
$\Theta_{\Lambda}$ on the subalgebra $\cA_\Lambda$. Note that $\cA_\Lambda^+$  is a $C^*$-subalgebra of $\cA_\Lambda$, but  $\cA_\Lambda^-$
is not a subalgebra. In fact, it is immediate that $(\cA_\Lambda^-)^2\subset \cA_\Lambda^+$.

A convenient basis for the local subalgebra $\cA_{\Lambda}$ is the one consisting of all monomials. Recall that 
$A \in \cA_{\Lambda}$ is a monomial if
\begin{equation} \label{monomials}
A = \prod_{x \in \Lambda} A_x \quad \mbox{with} \quad A_x \in \{ \idty, a_x, a_x^*, a_x^*a_x \} \, .
\end{equation}
Since each monomial is either even or odd, we conclude that any $A \in \cA_{\Lambda}$ can
be written as $A = A^++A^-$ where $A^{\pm} \in \cA_{\Lambda}^{\pm}$ and with each of $A^{+}$, resp. $A^{-}$, 
being a linear combination of even, resp. odd, monomials. Many of the results we will present
depend on the support of the observable under consideration. The following proposition describes some useful commutation properties
related to support.

\begin{prop} \label{prop:dis_sup} Let $X,Y \in \cP_0( \Gamma)$ with $X \cap Y = \emptyset$. 
\newline (i) If $A \in \cA_X^+$ and $B \in \cA_Y$, then $[A, B] = 0$. Moreover, if $A \in \cA_X$, $B \in \cA_Y$, and $[A,B] = 0$, then
either $A \in \cA_X^+$ or $B \in \cA_Y^+$.
\newline (ii) If $A \in \cA_X^-$ and $B \in \cA_Y^-$, then $\{A, B\} = 0$. Moreover, if $A \in \cA_X$, $B \in \cA_Y$, and $\{A,B\} = 0$, then
either $A$ or $B$ is identically zero, or  $A \in \cA_X^-$ and $B \in \cA_Y^-$.
\end{prop}
\begin{proof}
The first statements in (i) and (ii) above are easy to check for monomials, and they extend to general observables by linearity. Now, suppose
$A \in \cA_X$, $B \in \cA_Y$, and $[A,B] = 0$. Using the first part of (i), it is clear that
\begin{equation}
0 = [ A, B ] = [ A^+ + A^- , B^+ + B^-] = [A^-, B^-] = 2 A^-B^- - \{ A^-, B^-\}
\end{equation}
and therefore, $A^-B^- =0$ by the first part of (ii). The claim now follows by expanding $A^-$ and $B^-$ in the monomial basis. 
Proving the second part of (ii) is similar.
\end{proof}

An interaction $\Phi$ for a system of fermions on $\Gamma$ is defined similarly to that of an interaction for a quantum spin system.
Specifically, a map $\Phi : \mathcal{P}_0( \Gamma) \to
\mathcal{A}_{\Gamma}$ is an interaction if $\Phi(X)^* = \Phi(X) \in \cA_X$ for
all $X \in \mathcal{P}_0( \Gamma)$. For the results we are interested in here, we will restrict our attention to {\em even} interactions. An interaction $\Phi$
is said to be even if $\Phi(X) \in \cA_X^+$ for all $X\in \cP_0(\Gamma)$.
As such, each term $\Phi(X)$ is itself a sum of terms of the form
\be
a^*_Y h(Y,Z) a_Z + a^*_Z \overline{h(Y,Z)} a_Y, 
\ee
where $Y\cup Z=X$, $|Y|+|Z|$ is even, $h(Y,Z)\in\Cx$, and $a_Y=a_{y_1}\cdots a_{y_k}$ for a suitable enumeration of $Y=\{y_1,\ldots,y_k\}$. 
This assumption of even interactions has a physical justification; it follows from the conservation of angular momentum that the parity of the number of particles with half-integral spin is preserved. In other words, fermions can only be created or annihilated in an even number at a time.

Given an interaction $\Phi$ and $\Lambda \in \mathcal{P}_0( \Gamma)$, a local Hamiltonian 
$H_{\Lambda}^{\Phi}$ is defined by
\begin{equation} \label{local_ham}
H_{\Lambda}^{\Phi} = \sum_{X \subset \Lambda} \Phi(X).
\end{equation}
When the interaction $\Phi$ is understood, we often drop its dependence in the local Hamiltonians.
These finite-volume Hamiltonians generate a local Heisenberg dynamics which is a one-parameter group of automorphisms
of $\cA_\Lambda$:
\be
\tau^\Lambda_t(A) = e^{itH_\Lambda} A  e^{-itH_\Lambda}, \ t\in\Rl, A\in \cA_\Lambda.
\ee

If a dynamics $\tau^\Lambda_t$ is generated by an even interaction, it 
leaves $\cA_\Lambda^\pm$ invariant. In the next section we will show that for
such dynamics that Lieb-Robinson bounds identical to the well-known bounds for
quantum spin systems hold with essentially the same proof.
One could also consider lattice fermion systems with an infinite number of bands, but for simplicity we will not do this 
here. In the finite band case, it is not a loss of generality to assume that the time-dependence of each interaction 
term is continuous in the operator norm, and this allows for a more straightforward presentation. 

\section{Lieb-Robinson bounds}\label{sec:LRbounds}

Lieb-Robinson bounds provide an upper bound on the speed of propagation of disturbances in an extended many-body system.
Such bounds can be proved under quite general conditions on the many-body interactions. In fact, the argument we describe 
below applies to time-dependent interactions as well, and so we introduce this now. Let $I \subset \mathbb{R}$ be an interval;
often $I = \mathbb{R}$, but this is not necessary. An even, time-dependent interaction $\Phi$ is a mapping
$\Phi : \mathcal{P}_0( \Gamma) \times I \to \mathcal{A}_{\Gamma}$ for which
\begin{equation}
\Phi(X,t)^* = \Phi(X,t) \in \mathcal{A}_X^+ \quad \mbox{for all } X \in \mathcal{P}_0( \Gamma) \mbox{ and } t \in I,
\end{equation}
and moreover, 
\begin{equation}
t \mapsto \Phi(X,t) \mbox{ is continuous for each } X \in \mathcal{P}_0( \Gamma) \, .
\end{equation}

Associated to any even, time-dependent interaction $\Phi$ and each $\Lambda \in \mathcal{P}_0(\Gamma)$, there is
a corresponding finite-volume, time-dependent Hamiltonian 
\begin{equation} \label{fv_ham}
H_{\Lambda}(t) = \sum_{X \subset \Lambda} \Phi(X,t) \quad \mbox{for all } t \in I \, .
\end{equation}
As is well-known, see e.g. Theorem X.69 in \cite{reed:1975}, the solution of
\begin{equation} \label{uprop}
\frac{d}{dt} U_{\Lambda}(t,s) = -iH_{\Lambda}(t) U_{\Lambda}(t,s) \quad \mbox{with } U_{\Lambda}(s,s) = \idty \quad \mbox{and } s,t \in I, 
\end{equation}
produces a two-parameter family of unitary propagators $U_{\Lambda}(t,s) \in \mathcal{A}_{\Lambda}$, and in
terms of these unitaries, a co-cycle of automorphisms of $\mathcal{A}_{\Lambda}$ is defined by
\begin{equation} \label{fv_dyn}
\tau_{t,s}^{\Lambda}(A) = U_{\Lambda}(t,s)^* A U_{\Lambda}(t,s) \quad \mbox{for all } A \in \mathcal{A}_{\Lambda} \mbox{ and } t,s \in I \, .
\end{equation}
These automorphisms $\tau_{t,s}^{\Lambda}$ are commonly referred to as the finite-volume
Heisenberg dynamics associated to $\Phi$, and it is to these dynamics that the Lieb-Robinson bounds apply.
As we previously observed in the time-independent case, since $H_\Lambda(t)$ in
\eq{uprop} is even, $U_\Lambda(t,s)$ is even and so the map
$\tau_{t,s}^{\Lambda}$ leaves $\mathcal{A}_{\Lambda}^{+}$ and
$\mathcal{A}_{\Lambda}^{-}$ invariant.

Lieb-Robinson bounds are valid for interactions that decay sufficiently fast. 
A precise formulation of these bounds requires a notion of distance.
For this reason, we will further assume that the countable set $\Gamma$ is equipped with 
a metric $d$. In many physically interesting models, $\Gamma =\Ir^\nu$ for some integer $\nu \geq 1$, and we may, for instance,
take $d$ to be the $\ell^1$-metric. 

It is convenient to express the decay of interactions in terms of a function $G: \Gamma \times \Gamma \to (0,\infty)$ with the following properties:

(i) for all $x,y\in\Gamma$, $G(x,y) = G(y,x)$;

(ii) for all $x,y\in\Gamma$, $\sum_{z\in \Gamma} G(x,z)G(z,y) \leq G(x,y)$;

(iii) $x\mapsto \sum_{z\in \Gamma} G(x,z)$ is a uniformly bounded function on $\Gamma$.

\noindent We will denote the supremum of the bounded function described in (iii) by $\Vert G\Vert$.

Given such a function $G$, we define a set of even interactions, denoted
$\mathcal{B}_G^+(I)$, for which we can prove a Lieb-Robinson bound: an even interaction $\Phi$ belongs to $\mathcal{B}_G^+(I)$ if there is a locally integrable
function $\| \Phi \|_G : I \to [0, \infty)$ for which
\be
\sum_{\substack{Z\in\cP_0(\Gamma)\\ x,y\in Z}} \Vert \Phi(Z,t)\Vert \leq \Vert \Phi\Vert_G(t) G(x,y) \quad \mbox{ for all } x,y\in\Gamma \mbox{ and } t \in I.
\label{normPhi}\ee
Here $\Vert \Phi\Vert_G(t)$ plays the role of a time-dependent norm on the space of interactions. 

Two further definitions are useful in the statement of the Lieb-Robinson bounds. 
For $X \subset \Lambda \subset \Gamma$, we define a collection of surface sets associated to $X$ in $\Lambda$ by
\begin{equation}
S_{\Lambda}(X) = \{ Z \in \cP_0(\Lambda) : Z \cap X \neq \emptyset \mbox{ and } Z \cap ( \Lambda \setminus X) \neq \emptyset \}.
\end{equation}
Furthermore, if $\Phi$ is a time-dependent interaction, then the $\Phi$-boundary of a set $X \in \cP_0(\Gamma)$ is defined by
\be
\partial_\Phi X = \{ x\in X\mid \exists  Z\in S_\Gamma(X), t\in I \mbox{ s.t. } x\in Z, \mbox{ and } \Phi(Z,t) \neq 0\}.
\label{Phiboundary}\ee
If $\Phi$ is long-range, then it is often the case that $\partial_\Phi X = X$, and $\vert\partial_\Phi X\vert$ is not a good measure of
surface effects on the dynamics. However, when $\Phi$ is finite-range, the above definition is of more consequence.  

 \begin{theorem}[Lieb-Robinson Bound for Fermions] 
 \label{thm:lrb_f}
 Let $\Gamma$ be a countable set equipped with a metric $d$ and a function $G$ satisfying (i) - (iii) above,
and let $\Phi \in \cB_G^+(I)$. 
For $\Lambda\in\cP_0(\Gamma)$, let $\tau_{t,s}^\Lambda$ be a finite-volume dynamics associated to $\Phi$, 
 as defined in (\ref{fv_dyn}). Let $X,Y \subset \Lambda$ with $X \cap Y = \emptyset$. 
 \newline (i) If $A \in \cA_X$, $B \in \cA_Y$, and $[A,B]=0$, then 
 \begin{equation} \label{lrb_c_est}
 \left\| \left[ \tau_{t,s}^{\Lambda}(A), B \right] \right\| \leq 2 \| A \| \| B \| \left( \exp \left[ 2  \int_s^t \Vert\Phi\Vert_G(r) \, dr \right] - 1 \right) \sum_{x \in \partial_{\Phi} X} \sum_{y \in Y} G(x,y)
 \end{equation}
for all $t,s \in I$, $s\leq t$. 
\newline (ii) If $A \in \cA_X$, $B \in \cA_Y$, and $ \{ A,B \}=0$, then for any $s\leq t$, one has that
 \begin{equation} \label{lrb_ac_est}
 \left\| \left\{ \tau_{t,s}^{\Lambda}(A), B \right\} \right\| \leq 2 \| A \| \| B \| \left( \exp \left[ 2 \int_s^t \Vert\Phi\Vert_G(r) \, dr \right] - 1 \right) \sum_{x \in \partial_{\Phi} X} \sum_{y \in Y} G(x,y).
 \end{equation}
 \end{theorem}
 
 A number of remarks are in order.
 
First, the proof of this theorem mimics techniques that are well-known in the
context of quantum spin systems, i.e. tensor product algebras. In fact, 
given Lemma~\ref{lem:normbd} and Lemma~\ref{lem:normpresapp}, both proved below, 
the proof of Lieb-Robinson bounds in both settings proceeds almost identically.
This was noted already in \cite{bru:2017}, where a variant of  Theorem
\ref{thm:lrb_f} (i) is proved along the same lines as what we present here
(see \cite{bru:2017}[Theorem 5.1-Corollary 5.2(ii)]. It has been known for some
time that to treat odd observables for fermionic systems, one should consider anti-commutators in addition to 
commutators. See, e.g., the discussion in \cite{hastings:2006}.

Next, as explained in Proposition~\ref{prop:dis_sup}, the assumptions on
the observables above can be re-formulated: for (\ref{lrb_c_est}), one has assumed that
either $A$ or $B$ is even, whereas for (\ref{lrb_ac_est}), excepting trivialities, one 
has assumed that both $A$ and $B$ are odd. 

Further, the existence of a function $G$ satisfying properties (i)-(iii) above is, implicitly, a condition on $\Gamma$.
In many applications, the function $G$ is defined as a function of the distance between sites. 
More precisely, let $F:[0,\infty)\to(0,\infty)$ be a non-increasing function with the following two properties:

(iv) $F$ is uniformly integrable on $\Gamma$, i.e.
\begin{equation} \label{F:int}
\| F \| = \sup_{x \in \Gamma} \sum_{y \in \Gamma} F(d(x,y)) < \infty,
\end{equation}
and

(v) $F$ satisfies the convolution condition
\begin{equation} \label{F:conv}
C = \sup_{x,y \in \Gamma} \sum_{z \in \Gamma} \frac{F(d(x,z))F(d(z,y))}{F(d(x,y))} < \infty.
\end{equation}
Any function $F$ as above is called an $F$-function on $\Gamma$. Moreover, it is clear that, in terms of any such $F$, we can define a function 
$G$ with the properties (i)-(iii), by setting $G(x,y) = C^{-1} F(d(x,y))$.

For technical estimates, it is often convenient to consider classes of decay functions on $\Gamma$. Note 
that if $F$ is an $F$-function on $\Gamma$, then
for any subadditive function $f:[0,\infty)\to [0,\infty)$, i.e., $f(r+s) \leq f(r) + f(s)$ for all $r,s\in [0,\infty)$,
the function $F_f(r) = e^{-f(r)} F(r)$ also satisfies (iv) and (v) with $\| F_f \| \leq \| F \|$ and $C_f \leq C$. 
Similarly, given $G$ satisfying (i)-(iii) and any $g:\Gamma\to (0,1]$, the function $G_g(x,y) = g(x) g(y) G(x,y)$ also satisfies (i)-(iii). 
In this latter case, the function $g$ can be used to introduce a spatial dependence
in the decay of the interaction.

For $\Gamma = \mathbb{Z}^{\nu}$ and $d(x,y) = |x-y|$, i.e. the
$\ell^1$-distance, a typical example of an $F$-function is given by  
\begin{equation}
F(r) = \frac{1}{(1+r)^{\nu + \epsilon}}
\end{equation}
where $\epsilon >0$ can be arbitrary. In fact, it is clear that this $F$-function is uniformly integrable, i.e. (\ref{F:int}) holds, and
moreover, for (\ref{F:conv}) one may take
\begin{equation}
C = 2^{\nu + \epsilon} \| F \| \, .
\end{equation}
In combination with $f(r) = a r$, $a> 0$, we obtain a useful family of
$F$-functions $F_a$ given by $F_a(r)=e^{-ar}/(1+r)^{\nu+\epsilon}$.

Before proving Theorem~\ref{thm:lrb_f}, we state two simple lemmas. 
First, in Lemma~\ref{lem:normbd} below, we prove a basic estimate for solutions of certain 
$\cB(\cH)$-valued differential equations. Next, Lemma~\ref{lem:normpresapp} summarizes an application of Lemma~\ref{lem:normbd} 
demonstrating a one-step locality estimate for the dynamically evolved quantities of interest. 
Theorem~\ref{thm:lrb_f} will then follow by iterating the result from
Lemma~\ref{lem:normpresapp}.

We start with the following solution estimate.

\begin{lemma} \label{lem:normbd}  Let $\cH$ be a complex Hilbert space, $I \subset \mathbb{R}$ an interval, and 
$A,B: I \to \mathcal{B}(\mathcal{H})$, 
be norm continuous with $A$ self-adjoint, i.e. $A(t)^*=A(t)$ for all $t\in I$.
Then, for any $t_0 \in I$, the solution of the initial value problem 
\begin{equation} \label{gende}
\frac{d}{dt} f(t) = i[A(t), f(t)] + B(t) \quad \mbox{with} \quad f(t_0) = f_0 \in \mathcal{B}(\mathcal{H})
\end{equation}
satisfies the estimate
\begin{equation} \label{normbd0}
\| f(t) \| \leq \| f(t_0) \| + \int_{t_-}^{t_+} \| B (s) \| \, ds \quad \mbox{for any } t \in I.
\end{equation}
Here we have set $t_+ = \max \{ t, t_0 \}$ and $t_- = \min \{ t,t_0 \}$.
\end{lemma}

The assumption of norm continuity above is convenient because it guarantees that the mapping
$s \mapsto \| B(s) \|$ is continuous and thus measurable. A variant of this result
in the case that $A$ and $B$ are merely strongly continuous is proved in
\cite{nachtergaele:2014}.

\begin{proof} 
Since $A$ is self-adjoint, the unitary propagator corresponding to
\be
\frac{d}{dt}U(t,t_0) = iA(t)U(t,t_0) \quad \mbox{with } U(t_0, t_0) =\idty 
\ee
is well-defined for any $t \in I$. In this case, the function $g: I \to \cB(\cH)$ given by
\be
g(t) = U(t,t_0) g_0 U(t,t_0)^* 
\ee
is the unique solution of the initial value problem
\begin{equation} \label{npde}
\frac{d}{dt} g(t) = i [ A(t), g(t) ] \quad \mbox{with} \quad g(t_0) = g_0 \in \cB(\cH)\, .
\end{equation}
{F}rom this, one readily checks that
\begin{equation} \label{deff}
f(t)  = U(t, t_0) \left( f_0 + \int_{t_0}^{t} U(s, t_0)^* B(s) U(s, t_0) \, ds \right) U(t, t_0)^*
\end{equation} 
is the unique solution of (\ref{gende}) from which (\ref{normbd0}) follows.
\end{proof}

We now use Lemma~\ref{lem:normbd} to provide estimates on two 
families of operators. Let $\Phi$ be an even, time-dependent interaction.
For any $\Lambda \in \cP_0( \Gamma)$, consider the 
dynamics $\tau^\Lambda_{t,s}$ on $\cA_\Lambda$, as in (\ref{fv_dyn}) above. 
Given $X \subset \Lambda$ and any $B\in \cA_\Lambda$, we define maps $g_{t,s}^{X,B}, h_{t,s}^{X,B} : \cA_X\to\cA_\Lambda$
as follows:
\be
g_{t,s}^{X,B}(A) = [\tau^\Lambda_{t,s}(A), B] \quad \mbox{and} \quad h_{t,s}^{X,B}(A) = \{\tau^\Lambda_{t,s}(A), B\} \quad \mbox{for all } A \in \cA_X \, .
\label{mappings_gh}\ee

\begin{lemma}\label{lem:normpresapp}  Let $\Phi$ be an even, time-dependent interaction and
take $\Lambda \in \mathcal{P}_0( \Gamma)$.
For $X \subset \Lambda$, $A\in\cA_X$, and $s,t \in I$, with $s\leq t$, the
mappings $g_{t,s}^{X,B}$ and $h_{t,s}^{X,B} : \cA_X\to\cA_\Lambda$, defined in \eq{mappings_gh}, satisfy the bounds
\bea \label{gennormbd}
\| g^{X,B}_{t,s}(A) \| &\leq& \| [\tau_{t,s}^X(A), B] \|  + 2 \Vert A\Vert \sum_{Z \in S_{\Lambda}(X)} \int_s^t \| g^{Z,B}_{r,s}( \Phi(Z,r)) \| \, dr \label{g_normbd}\\
\| h^{X,B}_{t,s}(A) \| &\leq& \| \{ \tau_{t,s}^X(A), B\} \|  + 2 \Vert A\Vert
\sum_{Z \in S_{\Lambda}(X)} \int_s^t \| g^{Z,B}_{r,s}( \Phi(Z,r)) \| \, dr, \label{h_normbd}
\eea
where $\tau_{t,s}^X(A)$ is the finite-volume dynamics associated to $H_X(t)$, see (\ref{fv_ham}) and (\ref{fv_dyn}).
\end{lemma}

As will be the case in our applications, these bounds are particularly useful when the first term on the right-hand-sides above vanish. 

\begin{proof}[Proof of Lemma~\ref{lem:normpresapp}]
Let $A \in \mathcal{A}_X, B\in \cA_\Lambda$ and $s\in I$ be fixed. To derive the two bounds in parallel, define two functions:
\be
f_1(t) = [ \tau_{t,s}^{\Lambda} \circ \hat{\tau}_{t,s}^X  (A) , B ] \quad \mbox{and} \quad 
f_2(t) = \{ \tau_{t,s}^{\Lambda} \circ \hat{\tau}_{t,s}^X  (A) , B \}.
\ee
Here, $\hat{\tau}_{t,s}^X$  denotes the inverse dynamics of the system
restricted to $X$,  i.e.
\begin{equation}
\hat{\tau}^X_{t,s}(A) = U_X(t,s) A U_X(t,s)^*,
\end{equation}
which is the inverse of (\ref{fv_dyn}). One readily checks that for $f_1$:
\begin{eqnarray}
\frac{d}{dt} f_1(t) & = & i \left[ \tau_{t,s}^{\Lambda} \left( \left[ H_{\Lambda}(t) - H_X(t), \hat{\tau}_{t,s}^X(A) \right] \right), B \right]  \nonumber \\
& = & i \sum_{Z \in S_{\Lambda}(X)} \left[ \left[ \tau_{t,s}^{\Lambda}( \Phi(Z,t)), \tau_{t,s}^{\Lambda} \circ \hat{\tau}_{t,s}^X(A) \right], B \right] \nonumber \\
& = & i \!\!\!\!\! \sum_{Z \in S_{\Lambda}(X)}\!\!\! \left[ \tau_{t,s}^{\Lambda}( \Phi(Z,t)), f_1(t) \right] - i\!\!\!\!\! \sum_{Z \in S_{\Lambda}(X)} \!\!\! 
\left[ \tau_{t,s}^{\Lambda} \circ \hat{\tau}_{t,s}^X(A), \left[ \tau_{t,s}^{\Lambda}( \Phi(Z,t)), B \right] \right].  \nonumber
\end{eqnarray}
In the first equality above, we use that the adjoints of the unitary propagators satisfy the adjoint of (\ref{uprop}). 
For the second, we use that ${\rm supp}(\hat{\tau}_{t,s}^X (A)) \subset X$, and for the
final equality we use the Jacobi identity. 

An almost identical calculation for $f_2(t)$ gives:
\begin{eqnarray}
\frac{d}{dt} f_2(t) & = & i \left\{ \tau_{t,s}^{\Lambda} \left( \left[ H_{\Lambda}(t) - H_X(t), \hat{\tau}_{t,s}^X(A) \right] \right), B \right\}  \nonumber \\
& = & i \sum_{Z \in S_{\Lambda}(X)} \left\{ \left[ \tau_{t,s}^{\Lambda}( \Phi(Z,t)), \tau_{t,s}^{\Lambda} \circ \hat{\tau}_{t,s}^X(A) \right], B\right\} \nonumber \\
& = & i\!\!\!\!\!\! \sum_{Z \in S_{\Lambda}(X)} \!\!\!\left[ \tau_{t,s}^{\Lambda}( \Phi(Z,t)), f_2(t) \right] - i\!\!\!\!\! \sum_{Z \in S_{\Lambda}(X)} 
\!\!\!\left\{ \tau_{t,s}^{\Lambda} \circ \hat{\tau}_{t,s}^X(A), \left[ \tau_{t,s}^{\Lambda}( \Phi(Z,t)), B \right] \right\}.  \nonumber
\end{eqnarray}
We note that the only change here is that, instead of the Jacobi identity for
commutators, we use the following identity which holds for any three elements
$x,y,z$ in an associative algebra:
\be
\{ [x,y],z\} - \{[z,x],y\} + [\{y,z\},x] = 0.
\label{antiJacobi}\ee

Both differential equations above are of the form required to apply Lemma~\ref{lem:normbd}.
The claimed bounds follow from this and the substitution $A \mapsto \tau_{t,s}^X(A)$.
\end{proof}

We can now present the proof of the Lieb-Robinson bounds for lattice fermions
for which the following notation will be useful: for $X,Y\subset\Gamma$, set
\be
\delta_Y(X)= \begin{cases}
 0 & \mbox{if } X\cap Y = \emptyset\\ 1 & \mbox{if }  X\cap Y \ne \emptyset.\end{cases}
\ee

\begin{proof}[Proof of Theorem~\ref{thm:lrb_f}]
We first prove \eq{lrb_c_est}. Lemma~\ref{lem:normpresapp} implies
\begin{equation} \label{base_est_com}
\| [\tau_{t,s}^{\Lambda}(A), B ] \|  \leq  \| [ \tau_{t,s}^X(A), B] \| + 2 \| A \| \sum_{Z \in S_{\Lambda}(X)} \int_s^t  
\| [ \tau_{r,s}^{\Lambda}( \Phi(Z,r)), B ] \| \, dr.
\end{equation}
Since the observables $A$ and $B$ have disjoint support and one of them is even, 
it is clear that $[ \tau_{t,s}^X(A), B] =0$. More generally, if $A \in \cA_Z^+$, then $\tau_{r,s}^Z(A) \in \cA_Z^+$
for all $r,s \in I$ and so the bound $ \| [ \tau_{r,s}^Z(A), B] \| \leq 2 \| A \| \| B \| \delta_Y(Z)$ follows. In this case,
by iterating \eq{base_est_com} $N \geq 1$ times, we obtain
\begin{equation}
\Vert [\tau_{t,s}^{\Lambda}(A),B]\Vert \leq 2 \Vert A\Vert \Vert B\Vert  \left( \delta_Y(X) +  \sum_{n=1}^N a_n(t) \right) + R_{N+1}(t) 
\end{equation}
where
\begin{eqnarray}\label{ant}
a_n(t) = 2^n \sum_{Z_1 \in S_{\Lambda}(X)} \sum_{Z_2 \in S_{\Lambda}(Z_1)} \cdots \sum_{Z_n \in S_{\Lambda}(Z_{n-1})}  \delta_Y(Z_n) \int_s^t \int_s^{r_1} \cdots \int_s^{r_{n-1}} \times \nonumber \\
\times \left( \prod_{j=1}^{n} \| \Phi (Z_j, r_j) \| \right)  dr_n dr_{n-1} \cdots d r_1 
\end{eqnarray}
and 
\begin{eqnarray}\label{def:rem}
R_{N+1}(t) = 2^{N+1} \sum_{Z_1 \in S_{\Lambda}(X)} \sum_{Z_2 \in S_{\Lambda}(Z_1)}  \cdots  \sum_{Z_{N+1} \in S_{\Lambda}(Z_N)} \int_s^t \int_s^{r_1} \cdots  \int_s^{r_{N}} \times \nonumber \\  \left( \prod_{j=1}^{N} \| \Phi (Z_j, r_j) \|  \right) \Vert [\tau_{r_{N+1},s}^{\Lambda}(\Phi(Z_{N+1},r_{N+1})),B]\Vert  dr_{N+1} d r_{N} \cdots dr_1.
\end{eqnarray}

The remainder term $R_{N+1}(t)$ is estimated as follows. 
First, we observe that
\begin{equation}
\Vert [\tau_{r_{N+1},s}^{\Lambda}(\Phi(Z_{N+1}, r_{N+1}) ),B]\Vert\leq 2 \| B \|\, \| \Phi(Z_{N+1},r_{N+1})\|  \, .
\end{equation}
Next, we note that the sums above are actually sums over chains of sets $Z_1,\cdots Z_{N+1}$
that satisfy $Z_1 \cap \partial_{\Phi}X \neq \emptyset$ and $Z_j \cap Z_{j-1} \neq \emptyset$ for $2 \leq j \leq N+1$.
As such, there are points $w_1, w_2, \cdots, w_{N+1} \in \Lambda$ with $w_1 \in Z_1 \cap \partial_{\Phi}X$ and
$w_j \in Z_j \cap Z_{j-1}$ for all $2 \leq j \leq N+1$. A simple upper bound on
these sums is then obtained by over counting:
\begin{equation} \label{sumbd}
\sum_{Z_1 \in S_{\Lambda}(X)} \sum_{Z_2 \in S_{\Lambda}(Z_1)}  \!\!  \cdots   \!\! \sum_{Z_{N+1} \in S_{\Lambda}(Z_N)}  \!\!\!\!  *\quad \leq 
\sum_{w_1 \in \partial_{\Phi}X}\sum_{w_2,\ldots, w_{N+2} \in \Lambda} \sum_{\stackrel{Z_1, \ldots, Z_{N+1} \subset \Lambda:}{w_k, w_{k+1} \in Z_k, k=1,\ldots, N+1}} \!\!\!\! *
\end{equation}
where $*$ denotes an arbitrary non-negative quantity. Note that we have also
used that the last set $Z_{N+1}$ must contain more than one point since $Z_{N+1} \in S_{\Lambda}(Z_N)$.
Now, from \eq{normPhi} we have that
\begin{equation} \label{intbd1}
\sum_{\stackrel{Z_k \subset \Lambda:}{w_k, w_{k+1} \in Z_k}} \| \Phi(Z_k, r_k)  \|\leq \| \Phi \|_G(r_k) G(w_k, w_{k+1})
\end{equation}
holds for each $1 \leq k \leq N+1$. We conclude that
\begin{eqnarray*}
R_{N+1}(t) & \leq & 2 \| B \| 2^{N+1} \int_s^t \!\cdots\! \int_s^{r_N} \!\!\! \sum_{w_1 \in \partial_{\Phi}X}\sum_{w_2,\ldots, w_{N+2} \in \Lambda} 
\sum_{\stackrel{Z_1, \ldots, Z_{N+1} \subset \Lambda:}{w_k, w_{k+1} \in Z_k, k=1,\ldots, N+1}}\\
&&\quad\times \prod_{j=1}^{N+1} \| \Phi(Z_j,r_j)  \| dr_{N+1} \cdots dr_1 \nonumber \\
& \leq &  2 \| B \| 2^{N+1} \int_s^t \cdots \int_s^{r_N} \sum_{w_1 \in \partial_{\Phi}X} 
\sum_{w_2,\ldots, w_{N+2} \in \Lambda}\\
&&\quad\times \prod_{j=1}^{N+1} \| \Phi \|_G(r_j) G(w_j,w_{j+1}) dr_{N+1} \cdots dr_1 \nonumber \\ 
& \leq & 2 \| B \| 2^{N+1}  \sum_{w_1 \in \partial_{\Phi}X}  \sum_{w_{N+2} \in \Lambda} G(w_1, w_{N+2})\\
&&\quad\times \int_s^t \cdots \int_s^{r_N} \prod_{j=1}^{N+1} \| \Phi \|_G(r_j)
dr_{N+1} \cdots dr_1 \nonumber \\ 
& \leq & 2 \| B \| | \partial_{\Phi} X| \| G \|  \frac{ \left( 2 \int_s^t \|
\Phi \|_G(r) \, dr \right)^{N+1}}{(N+1)!}.
\end{eqnarray*}
Since $\| \Phi \|_G(\cdot)$ is locally integrable on $I$, this remainder clearly goes to 0 as $N \to \infty$.

A similar estimate can be applied to the terms $a_n(t)$. Note that these terms
are also sums over chains of sets. However,
there is a restriction: only those chains whose final link $Z_n$ satisfies $Z_n \cap Y \neq \emptyset$ contribute to the sum. The bound
\begin{equation}
a_n(t) \leq  \frac{\left( 2  \int_s^t \| \Phi \|_G(r) \, dr \right)^n}{n!}
\sum_{x \in \partial_{\Phi}X} \sum_{y \in Y} G(x,y).
\end{equation}
follows as above. Since $\delta_Y(X) = 0$ and $n \geq 1$, the bound in (\ref{lrb_c_est}) is now clear.

The proof of \eq{lrb_ac_est} proceeds similarly.  In fact, Lemma~\ref{lem:normpresapp} implies
\begin{equation} \label{base_est_acom}
\| \{ \tau_{t,s}^{\Lambda}(A), B \} \|  \leq  \| \{ \tau_{t,s}^X(A), B \} \| + 2 \| A \| \sum_{Z \in S_{\Lambda}(X)} \int_s^t  
\| [ \tau_{r,s}^{\Lambda}( \Phi(Z,r)), B ] \| \, dr
\end{equation}
Since $\{A, B\} =0$, the first term on the left-hand-side above vanishes. Using \eq{base_est_com} to estimate the second term and iterating  
in exactly the same way as above, yields (\ref{lrb_ac_est}) as claimed. 
\end{proof}
 
Theorem~\ref{thm:lrb_f} gives an estimate of $\left\| \left[ \tau_{t,s}^{\Lambda}(A), B \right] \right\|$ for commuting observables $A$ and $B$ with 
disjoint supports. As a function of $|t-s|$, this estimate grows exponentially and for small $|t-s|$ vanishes linearly. A few additional comments are in 
order. First, note that when the supports of $A$ and $B$ have non-empty intersection, in general, one cannot
 expect to improve on the trivial bound:  $\left\| \left[ \tau_{t,s}^{\Lambda}(A), B \right] \right\| \leq 2 \Vert A\Vert \Vert B\Vert$. 
 On the other hand, given a lower bound on the distance between the supports of $A$ and $B$, and if the interaction is
of finite range, a slight modification to the proof of Theorem~\ref{thm:lrb_f}
shows that the behavior for small $|t-s|$ is $o(|t-s|^n)$, where $n$ is the
minimum number of interactions terms necessary to connect the supports of $A$
and $B$ with a chain of sets, see e.g. (\ref{ant}). For similar reasons, one can
also show that single-site terms in the Hamiltonian do not contribute to the estimate of Lieb-Robinson velocity. This can easily be seen using the interaction picture as is done, e.g., in \cite{nachtergaele:2009b}. Of course, this is not the statement that single-site terms in the Hamiltonian do not affect the velocity associated to certain time-evolved observables; rather it is the fact that this general upper bound is insensitive to such terms. For example, in specific models with a random external field, the speed of propagation has been shown to vanish \cite{hamza:2012}.

\section{Conditional expectation and local approximations}\label{sec:cond_exp}

In many applications to quantum spin systems, the commutator estimates provided
by Lieb-Robinson bounds are used to approximate quasi-local observables by
strictly local ones. These local approximations are given by a conditional
expectation with respect to a suitable product state (see, e.g.,
\cite{bravyi:2006a,nachtergaele:2006,nachtergaele:2013}). In the setting of
lattice fermions, the conditional expectations that come to mind are those with 
respect to a product state such as, e.g., the tracial state, which is also the
quasi-free state on $\cA_\Gamma$ determined by $\omega^{\rm tr}(a_x a_y)=0$,
and  $\omega^{\rm tr}(a^*_x a_y)=\frac{1}{2}\delta_{x,y}$, for all $x,y\in
\Gamma$. This state has the following product property \cite[Theorem
6.12]{alicki:2001}: for any finite set of distinct $x_1,\ldots,x_k \in\Gamma$,
and $A_{x_i}\in \cA_{\{x_i\}}$, we have
\be\label{productproperty}
\omega^{\rm tr}(A_{x_1}\cdots A_{x_k}) = \prod_{i=1}^k \omega^{\rm tr}(A_{x_i}).
\ee
It is not difficult to see that for all $\Lambda\in\cP_0(\Gamma)$, $\omega^{\rm tr}$ restricted to
$\cA_\Lambda$ is the state of maximal entropy. It was shown in \cite{araki:2003}[Theorem 4.7] that 
for each finite $X\subset\Gamma$, there is a unique conditional expectation $\bF_X:\cA_\Gamma\to\cA_X$
that leaves $\omega^{\rm tr}$ invariant, meaning
\be
\omega^{\rm tr}(\bF_X(AB)) =\omega^{\rm tr}(\bF_X(A)B), \quad A\in\cA_\Gamma, B\in \cA_X.
\label{CEaraki}\ee
The same result for arbitrary even product states $\omega$ was proved in \cite{araki:2004}.
The family of conditional expectations 
$\{\bF_X: \cA_\Gamma\to\cA_X \mid X\in \cP_0(\Gamma)\}$ determined by \eq{CEaraki}, 
satisfies the commuting diagram:
\be
\begin{tikzcd}
  \cA_{X\cup Y}  \arrow[r, "\bF_X"] \arrow[d, "\bF_Y"]
    & \cA_X \arrow[d, "\bF_{X\cap Y}"] \\
  \cA_Y \arrow[r, "\bF_{X\cap Y}"]
&\cA_{X\cap Y} \end{tikzcd}.
\label{commutingsquare}\ee
The commutativity of this diagram is the essential property that allows one
to use the maps $\bF_X$ in the same role as the partial trace in the case of 
of quantum spin systems. We note however, that for applications one often only needs 
these relations for even observables, in which case there are other options for the local algebras
and the conditional expectations. For our purposes, it will be useful to introduce a slightly different 
set of conditional expectations, which we will now explain.

Under the isomorphism $\cA_\Lambda\cong M_{2^{|\Lambda|}}$, $\omega^{\rm tr}$ coincides with 
the normalized trace on $M_{2^{|\Lambda|}}$. For even observables, the conditional expectation
acts exactly as in the situation of spin system. The anticommutation properties of odd observables, 
however, introduce a small twist, which implies that the Krauss form of the map $\bF_X$ contains 
global operators (see \eq{nonlocalkrauss} below). This creates a complication for the proof of Lemma \ref{lem:local_approx_f}.
Therefore, we introduce another family of conditional expectations which, as we will show, also
satisfy a commuting diagram similar to \eq{commutingsquare}, and which coincide with the maps $
\bF_X$ on even observables.

The remainder of this section is devoted to describing the conditional expectations of interest for finite volume systems.
Given a $C^*$-algebra $\cA$, and a subalgebra $\cB$ of $\cA$,
by a conditional expectation of $\cA$ onto $\cB$, we mean a unity-preserving, completely positive map
$\bE: \cA \to \cA$ with $\ran(\bE)=\cB$, such that $\bE(BAC) = B\bE(A)C$ for all $A\in \cA$ and $B, \,
C \in \cB$. Let $X \subseteq \Lambda \subset \Gamma$ with $X$ and
$\Lambda$ finite.  We will show that the range of the relevant conditional
expectation $\bE_X^\Lambda : \cA_\Lambda \to \cA_\Lambda$ is given by the $C^*$
-subalgebra:
\[\cA_X^\Lambda = \{ A + B \theta_\Lambda \, : \, A\in \cA_X^+, \, B\in
\cA_X^-\},\]
where $\theta_\Lambda=(-1)^{N_\Lambda}$ is the parity operator
used to define $\Theta_\Lambda$ in \eqref{Parity_Op}.
Since we only consider Hamiltonians defined by even interactions, in applications
it will be sufficient to only consider the restriction of $\bE^\Lambda_X$ to
$\cA_\Lambda^+$. In the restricted case $\bE_X^\Lambda:
\cA_\Lambda^+ \to \cA_\Lambda$ will be a conditional expectation with
range equal to $\cA^+_X$.

We define $\bE_X^\Lambda$ by giving its Krauss form. For each site $x\in
\Lambda$, define
\be
u^{(0)}_x = \idty,
\quad 
u^{(1)}_x = a^*_x + a_x,
\quad 
u^{(2)}_x = a^*_x - a_x,
\quad 
u^{(3)}_x =  \idty - 2 a^*_x a_x.
\ee
It follows from the CAR that these are unitary. Clearly, $u^{(i)}_x \in
\cA_{\{x\}}^+$ for $i = 0,\, 3$, and  $u^{(i)}_x \in \cA_{\{x\}}^-$ for $i = 1,
\,2$. Therefore, $u^{(i)}_x$ commutes with the elements of
$\cA_{\Lambda\setminus \{x\}}$ for $i = 0, \, 3$, and $u^{(i)}_x$ commutes with
$\cA_{\Lambda\setminus \{x\}}^+$ and anti-commutes with $\cA_{\Lambda\setminus
\{x\}}^-$ for $i = 1, \, 2$.

For a subset $X\subseteq \Lambda$, let $I_{\Lambda\setminus X} =
\{0,1,2,3\}^{\Lambda \setminus X}$ and fix an ordering of the sites of
$\Lambda \setminus X = \{x_1,\ldots,x_{n} \}$. Then, for each $\alpha \in
I_{\Lambda \setminus X}$ define the unitary operator $u(\alpha)\in \cA_{\Lambda
\setminus X}$ by
\be
u(\alpha) = u_{x_1}^{(\alpha(x_1))}\cdots u_{x_n}^{(\alpha(x_{n}))}.
\ee
We will show that the following unity-preserving, completely positive map
$\bE^\Lambda_X:
\cA_\Lambda \to \cA_\Lambda$ is the conditional expectation of
interest associated with $X$ on $\Lambda$: 
\be\label{def_ELambda}
\bE_X^\Lambda(A) = \frac{1}{4^{|\Lambda\setminus X|}} \sum_{\alpha \in I_{\Lambda\setminus X}} u(\alpha)^* A u(\alpha),\quad A\in \cA_\Lambda.
\ee
It is important to notice that the map $\bE^\Lambda_X$ does not depend on the
ordering chosen for the site in $\Lambda \setminus X$. Since the unitaries
$u_x^{(k)}, u_y^{(l)}$, $x\neq y$, $k,l\in\{0,1,2,3\}$, either commute or
anti-commute, any reordering $\tilde u(\alpha)$ of $u(\alpha)$ equals either
$u(\alpha)$ or $- u(\alpha)$. Either way, the $\alpha$-term in \eq{def_ELambda}
is not affected.

The first important property for $\bE_X^\Lambda$ is that \be
\label{TraceProperty} \bE_X^\Lambda(A) = \omega^\tr(A)\idty \;\text{  for all }
\;A\in \cA_{\Lambda\setminus X}.
\ee
It is easy to verify that $\bE_{\Lambda \setminus \{x\}}^\Lambda(A) =
\omega^{\tr}(A)\idty$ for a monomial $A\in \cA_{\{x\}}$. Using this, the product
property of $\omega^\tr$ given in \eqref{productproperty}, and the CAR it
follows that \eqref{TraceProperty} holds for any monomial $A\in \cA_{\Lambda
\setminus X}$. The property then extends to $\cA_{\Lambda \setminus X}$ by
linearity.

To establish that $\bE_X^\Lambda$ is a conditional expectation, it is left to
verify that
\be\label{CondExp}
\bE_X^\Lambda(BAC) = B\bE_X^\Lambda(A)C
\ee
for all $A\in \cA_\Lambda$ and $B, \, C \in \ran(\bE_X^\Lambda)$. In the
situation that $\ran(\bE_X^\Lambda)$ is a $C^*$-subalgebra of $\cA_\Lambda$,
a theorem by Tomiyama  \cite{tomiyama:1957} shows that \eqref{CondExp} is satisfied
if $\bE_X^\Lambda$ is a norm-1 projection. In the next lemma we establish that
Tomiyama's result applies to both $\bE_X^\Lambda:\cA_\Lambda \to \cA_\Lambda$
and the restriction $\bE_X^\Lambda:\cA_\Lambda^+ \to \cA_\Lambda$.

\begin{lemma}
For $X, \, \Lambda \in \cP_0(\Gamma)$ with $X \subset \Lambda$, the map $\bE^\Lambda_X: \cA_\Lambda \to \cA_\Lambda$ defined in \eq{def_ELambda} satisfies

(i) $\Vert \bE^\Lambda_X\Vert =1 $;

(ii) $ (\bE^\Lambda_X)^2=  \bE^\Lambda_X$;

(iii) $\bE^\Lambda_X(\cA_\Lambda) = \cA_X^\Lambda$;

(iv) $\bE^\Lambda_X(\cA_\Lambda^+) = \cA_\Lambda^+$.
\end{lemma}

\begin{proof}
Property (i) follows immediately from the fact that $\bE_X^\Lambda$ is an
average of unitary conjugations. 

To prove (ii)-(iv) we consider observables of the form $A = \prod_{x\in
\Lambda}A_x$ where $A_{x}\in
\cA_{\{x\}}^{\pm},$ for all $x\in\Lambda$. We refer to such operators as
product observables. Since the monomials defined in \eqref{monomials} satisfy this
condition, computing $\bE_X^\Lambda(A)$ for product observables completely
describes the map.

By (anti-)commuting the factors $A_{x}$, $A$ can be written in the form $A=BC$
where $B$ and $C$ are product observables in $\cA_X$ and $\cA_{\Lambda\setminus
X}$, respectively. Then, using the CAR:
\be
\bE^\Lambda_X(A)= \begin{cases}
B\frac{1}{4^{|\Lambda\setminus X|}} \sum_{\alpha \in
I_{\Lambda\setminus X}}  u(\alpha)^* C  u(\alpha) & \mbox {if } B\in
\cA_\Lambda^+\\
B\frac{1}{4^{|\Lambda\setminus X|}} \sum_{\alpha \in
I_{\Lambda\setminus X}} \pi(u(\alpha)) u(\alpha)^* C  u(\alpha) & \mbox {if }
B\in \cA_\Lambda^-.
\end{cases}
\label{on_monomials}\ee
Here, $\pi(u(\alpha)) = \pm 1$ denotes the parity of $u(\alpha)$. For all
$\alpha \in I_{\Lambda \setminus X}$ 
\[\pi(u(\alpha)) u(\alpha)^* C  u(\alpha) 
	=
u(\alpha)^*C\theta_{\Lambda\setminus X}u(\alpha)\theta_{\Lambda\setminus X},\]
so applying \eqref{TraceProperty} to \eqref{on_monomials} we find:
\be\label{ProductOps}
\bE^\Lambda_X(A)= \begin{cases}
B \omega^{\tr}(C) & \mbox {if } B\in \cA_\Lambda^+\\
B \omega^{\tr}(C\theta_{\Lambda\setminus X})\theta_{\Lambda \setminus X} &
\mbox {if } B\in \cA_\Lambda^-.
\end{cases}
\ee
Applying \eqref{ProductOps} a second time shows
$(\bE_{X}^\Lambda)^2(A) = \bE_X^\Lambda(A)$ for any product observable $A$, and
so property (ii) holds. Alternatively, (ii) also follows from the observation that 
$\bE_{X}^\Lambda$ is defined as the average of the adjoint actions 
$\{{\rm Ad}_{u(\alpha)}\mid \alpha \in I_{\Lambda\setminus X}\}$, which form a group.

For property (iii), we see from \eqref{ProductOps} that
$\bE_X^\Lambda(\cA_\Lambda) \subseteq \cA_X^\Lambda$ by noting that
$\theta_{\Lambda\setminus X}=\theta_X\theta_\Lambda$. The opposite containment
follows from observing that
\[\bE_X^\Lambda(A+B\theta_\Lambda) = A+B\theta_\Lambda\]
for any pair of product observables $A\in \cA_X^+$ and $B\in \cA_X^-$.

For (iv), if $A\in \cA_\Lambda^+$ then the factors $B$ and $C$ above are either
both even or both odd. If they are both odd, then $\bE_X^\Lambda(A) = 0$
since $\omega^\tr$ is zero for odd observables. It follows that
$\bE_X^\Lambda(\cA_\Lambda^+) \subseteq \cA_X^+$. The opposite containment
holds since $\bE_X^\Lambda(A) = A$ for all $A\in \cA_X^+$.
\end{proof}

Several comments are in order. First, an important consequence of
\eqref{ProductOps} is that
\be
\cA_{\Lambda \setminus X}^- \subseteq \ker(\bE_X^\Lambda).
\ee

Second, we will want to consider the family of all
conditional expectations $\bE_X^\Lambda$ such that $X\subseteq \Lambda \in
\cP_0(\Gamma)$.
Since the definition of $\bE_X^\Lambda$ is independent of the ordering of the sites $x\in
\Lambda \setminus X$, given any two sets $Y, Z \subseteq \Lambda$ such
that $\Lambda \setminus X$ is the disjoint union of $\Lambda \setminus Y$ and
$\Lambda \setminus Z$ it immediately follows that
\be\label{disjoint_decomp}
\bE_X^\Lambda = \bE_Y^\Lambda \circ \bE_Z^\Lambda 
	= \bE_Z^\Lambda \circ \bE_Y^\Lambda.
\ee
%
\begin{lemma}\label{lem:properties_ELambda}
The family of conditional expectations $\{\bE_X^\Lambda \, : \, X \subseteq
\Lambda \in \cP_0(\Gamma)\}$ defined as in \eqref{def_ELambda} satisfy the
following properties:\\
(i) For any $X, \, Y \subseteq \Lambda$,
\be
\bE_{X}^\Lambda \circ \bE_Y^\Lambda = \bE_{X\cap Y}^\Lambda.
\ee
(ii) For any $X\subseteq \Lambda$ and $A\in \cA_Z^+$ and $B\in \cA_Y^+$ with $Z
\cap Y = \emptyset$,
\be\label{ProductProp}
\bE_X^\Lambda(AB) \,=\, \bE_{X\cup Y}^\Lambda(A) \cdot \bE_{X\cup
Y^c}^\Lambda(B)   \,=\, \bE_{X\cup Z^c}^\Lambda(A) \cdot \bE_{X\cup
Z}^\Lambda(B).
\ee
(iii) Given $X\subseteq \Lambda_1 \subseteq \Lambda_2$ and $A\in
\cA_{\Lambda_1}^+$,
\be
\bE_X^{\Lambda_1}(A) = \bE_X^{\Lambda_2}(A).
\ee

\end{lemma}
\begin{proof}
For (i), since $\Lambda \setminus X$ is the disjoint union of $\Lambda
\setminus (X\cup Y)$ and $\Lambda \setminus (X\cup Y^c)$,
\eqref{disjoint_decomp} implies 
\be\label{XY_decomp}
\bE_X^{\Lambda} = \bE_{X\cup Y^c}^\Lambda \circ \bE_{X\cup Y}^\Lambda.
\ee
 Analogously, $\bE_Y^\Lambda = \bE_{X\cup Y}^\Lambda \circ
\bE_{X^c\cup Y}^\Lambda$. Using that $\bE_{X\cup Y}^\Lambda$ is a projection, we find 
\[\bE_X^\Lambda \circ \bE_Y^\Lambda = \bE_X^\Lambda \circ \bE_{X^c\cup
Y}^\Lambda.\] 
The result follows from noticing that $\Lambda \setminus (X\cap
Y)$ is the disjoint union of $\Lambda \setminus X$ and $\Lambda \setminus (X^c
\cup Y)$.

For (ii), we first use \eqref{XY_decomp} to rewrite $\bE_X(AB)$. Then, using
the commutation properties of even observables we find
\begin{eqnarray*}
\bE_{X}^\Lambda(AB) 
	& = &
\frac{1}{4^{|\Lambda \setminus X|}}
\sum_{\substack{\alpha \in I_{X^c\cap Y^c} \\ \beta \in I_{X^c\cap Y}}}
u(\beta)^*u(\alpha)^*A u(\alpha) \, B u(\beta) \\
	& =  &
\frac{1}{4^{|\Lambda \setminus X|}}
\sum_{\substack{\alpha \in I_{X^c\cap Y^c} \\ \beta \in I_{X^c\cap
Y}}} u(\alpha)^*A u(\alpha) \, u(\beta)^* B u(\beta) \\
	& = &
\bE_{X\cup Y}^\Lambda(A) \cdot \bE_{X\cup Y^c}^\Lambda(B). 
\end{eqnarray*}
An analogous argument holds for showing the final equality in
\eqref{ProductProp}.

Finally, (iii), since $A' = \bE_X^{\Lambda_1}(A) \in \cA_X^+$, and
$\cA_X^+ \subseteq \cA_{\Lambda_1}^+\subseteq \cA_{\Lambda_2}^+$, it is clear that
\[\bE_X^{\Lambda_2}(A) = \bE_{\Lambda_1}^{\Lambda_2}(A') = A'.\]
\end{proof}

Recall that the motivation for introducing these conditional expectations was
to produce local approximations of global observables. Since we consider even
interactions, it is sufficient to just localize even observables. The final
result we provide shows that $\bE_X^\Lambda$ does indeed produce local
approximations of even observables.

\begin{lemma}\label{lem:local_approx_f}
Let $A\in \cA_\Lambda^{\rm +}$, $X\subset \Lambda$, and $\epsilon >0$.
If $\Vert [A,B]\Vert \leq \epsilon\Vert B\Vert $ for all $B\in\cA_{X^c}$,
then there exists $A^\prime\in \cA_X^+$ such that $\Vert A-A^\prime\Vert \leq
\epsilon$.
\end{lemma}

This lemma is proved by a straightforward estimate using $A^\prime=\bE_X^\Lambda(A)$:
\be
\Vert A-\bE_X^\Lambda(A)\Vert =\left\Vert  \frac{1}{4^{|\Lambda\setminus X|}} \sum_{\alpha \in I_{\Lambda\setminus X}} u(\alpha)^* [u(\alpha),A]\right\Vert .
\ee
Since $u(\alpha)\in \cA_{\Lambda\setminus X}$ and is unitary, for all $\alpha  \in I_{\Lambda\setminus X}$, by the assumption in the 
lemma we have the bound
\be
\Vert [A,u(\alpha)]\Vert \leq \epsilon,
\ee
which is sufficient to prove lemma.

Since $A$ is even, using the Krauss forms of $\bE^\Lambda_X$, see \eq{def_ELambda}, and of $\bF^\Lambda_X$, see \eq{nonlocalkrauss} below, 
one can easily check $\bE^\Lambda_X(A) = \bF^\Lambda_X(A)$. Note, however, that the Krauss form for $\bF^\Lambda_X$ on all of 
$\cA_\Lambda$ involves global operators and, as such, is not suitable for the above argument. To present a Krauss form for $\bF^\Lambda_X$, 
we first define $\pi(0)=\pi(3) = 1$, $\pi(1)=\pi(2)=-1$ and for $\alpha\in I_Y, \pi(\alpha)= \prod_{y\in Y} \pi(\alpha_y)$, and denote by $I_Y^\pm$ the sets of $\alpha \in I_Y$ for which $\pi(\alpha) = \pm 1$. Then we have
\be\label{nonlocalkrauss}
\bF_X^\Lambda(A) = \frac{1}{4^{|\Lambda\setminus X|}} \sum_{\alpha \in I_{\Lambda\setminus X}} \tu(\alpha)^* A \tu(\alpha),\quad A\in \cA_\Lambda,
\ee
with 
\be
\tu(\alpha) = \begin{cases} u(\alpha) &\mbox{if } \alpha \in I^+_{\Lambda\setminus X} \\
 \theta_X u(\alpha) &\mbox{if } \alpha \in I^-_{\Lambda\setminus X} 
\end{cases}\ .
\ee
Clearly, the support of $ \theta_X u(\alpha) $ is all of $\Lambda$ and there is no obvious way to apply a Lieb-Robinson bound to estimate 
a commuator with it.

\section{Martingale method for lattice fermions}\label{sec:martingale_method}

\newcommand{\cG}{\mathcal{G}}

Before we address the question of stability of a spectral gap in the presence of 
perturbations, we would like to demonstrate the existence of a wealth of models
for which one can prove the existence of a spectral gap. These models can then serve
as  the `unperturbed model.'

As in the situation of quantum spin systems, it is helpful to 
consider frustration-free models. Examples of frustration-free fermion models in one
dimension are easily found by making use of the Jordan-Wigner transformation
\cite{lieb:1961}, to translate the abundance of frustration-free quantum spin chains
into fermion Hamiltonians. An example of this type is Kitaev's Majorana chain \cite{kitaev:2001}
and also the path of frustration-free fermion chains connecting Kitaev's Majorana 
chain to a family of interacting fermion chains introduced by Katsura et al. \cite{katsura:2015}.

 In Section \ref{sec:stability} we will also introduce 
a class of frustration-free fermion models in arbitrary dimension that have a 
quasi-free ground state and a spectral gap above it. In the latter case, 
for the quadratic Hamiltonians associated to these models, one can often determine a lower
bound for the spectral gap by inspection. For models with interactions of a more
general form, however, proving a volume-independent lower bound for the spectral
gap above the ground state is generally a challenging problem. For spin systems
the martingale method has often been applied to establish a nonzero spectral gap
above the ground state energy
\cite{fannes:1992,nachtergaele:1996,bachmann:2015}.
In this section we give a formulation of the method suited for lattice fermions.
Although the tensor product structure of quantum spin systems is not available,
under general conditions, one can still obtain the commutation relations
required for the method. The approach we introduce here was first used in
\cite{young:2016}.

Let $\Lambda$ be finite and $H_\Lambda \in \cA_\Lambda^+$ be the lattice fermion
Hamiltonian acting on a fermionic Fock space.
Note further that for any self-adjoint, even observable $A\in\cA_\Lambda^+$
that the spectral projections associated to $A$ also belong to $\cA_\Lambda^+$. 
This follows immediately from the fact that the spectral projections can be 
expressed as a polynomial of $A$. {F}rom this, we conclude that the 
spectral projections of two self-adjoint even observables with disjoint support
commute.

Now consider an increasing sequence of non-negative Hamiltonians 
\[H_0 \leq H_1 \leq \ldots \leq H_N\]
with $H_0=0$, $H_N=H_\Lambda$, and $H_n \in \cA_\Lambda^+$ for all $ 0 \leq n \leq N$.
More generally, it is sufficient to consider an increasing, non-negative sequence of 
Hamiltonians for which
\be
c H_N \leq H_\Lambda - E_\Lambda\idty \leq C H_N
\ee
where $c,C>0$ and $E_\Lambda$ is the ground state energy of $H_\Lambda$. When
the conditions of the martingale method are met, it produces a lower bound for the
spectral gap above the ground state energy of $H_N$.

To apply the martingale method, it is necessary that the Hamiltonians $H_n$ each
have a non-trivial kernel. Since these Hamiltonians are
non-negative and increasing, the kernels $\cG_n =\ker H_n$ form a decreasing sequence of subspaces:
\[\cH_\Lambda=\cG_0\supset\cG_1\supset\cG_2 \cdots \supset \cG_N = \ker(H_N).\]

In order to state the assumptions of the martingale method we define 
\be
h_n = H_n - H_{n-1}, \text{ for } n=1,\ldots,N.
\label{defhn}\ee
Clearly, \eq{defhn} implies $H_n=\sum_{k=1}^n h_k$, and $H_n$ is increasing if and only if
$h_k\geq 0$, for all $k=1,\ldots, N$. Furthermore, define
$G_n$ and $g_n$ to be the orthogonal projections onto $\ker H_n$ and $\ker h_n$,
respectively, and let
\be
E_n=\begin{cases} \idty  - G_1 & \text{ if } n=0\\ G_n - G_{n+1} & \text{ if } 1\leq n \leq N-1\\ G_N & \text{ if } n=N\end{cases}.
\ee
It is easy to verify that the $E_n$ are a mutually orthogonal family of
orthogonal projections that form a resolution of the identity, i.e., $E_n^* = E_n$,
$E_n E_m = \delta_{n,m}E_n$, and $\sum_{n=0}^N E_n = \idty$.

\noindent
{\em Assumptions for the Martingale Method:} 

(i) There is a constant $\gamma>0$ such that $h_n\geq \gamma(\idty- g_n)$ for $1 \leq n \leq N$.

(ii) There is an integer $\ell\geq 0$ such that whenever $0 \leq n \leq N-1$ and 
 $k \not\in [n-\ell,n]$, $[E_k, g_{n+1}] = 0$.

(iii) There exists a positive $\epsilon < 1/\sqrt{\ell+1}$, such that $E_n g_{n+1}
E_n \leq \epsilon^2 E_n$, for $0 \leq n \leq N-1$.

\begin{theorem}[Martingale Method]
Suppose that Assumptions (i)--(iii) hold for a sequence of Hamiltonians $H_n$,
$n=0,\ldots, N$ as described above. If $\psi\in\cH_\Lambda$ such that
$G_N\psi=0$, then 
$$ \langle \psi, H_N \psi\rangle \geq \gamma
(1-\epsilon\sqrt{1+\ell})^2\Vert \psi\Vert^2.
$$
\end{theorem}
\begin{proof}
By assumption $E_N \psi = G_N\psi =0$. Hence 
\be
\Vert \psi\Vert^2 = \sum_{n=0}^{N-1} \Vert E_n \psi\Vert^2 .
\ee
Given this, for any $0 \leq n \leq N-1$, we have that 
\bea
\Vert E_n \psi\Vert^2 &=& \langle\psi, (\idty - g_{n+1})E_n\psi\rangle +  \langle\psi, g_{n+1}E_n\psi\rangle\nonumber\\
& = & \langle\psi, (\idty - g_{n+1})E_n\psi\rangle +  \langle\psi, \left(\sum_{k=n_{\ell}}^nE_k\right)g_{n+1}E_n\psi\rangle ,
\eea
where we have set $n_{\ell} = \max(0, n-\ell)$. For the last equality above, 
we inserted the resolution of the identity, used Assumption (ii), and applied the mutual orthogonality of the projections $E_n$.
Two applications of the inequality
\be
\left|\langle \varphi_1,\varphi_2\rangle\right| \leq \frac{1}{2c}\Vert \varphi_1\Vert^2 + \frac{c}{2}\Vert \varphi_2\Vert^2,\ \varphi_1,\varphi_2\in\cH, c>0,
\ee
now show that
\beann
\Vert E_n\psi \Vert^2 &\leq&
 \frac{1}{2c_1}  \langle\psi, (\idty - g_{n+1})\psi\rangle + \frac{c_1}{2}  \langle\psi, E_n\psi\rangle\\ 
&& +  \frac{1}{2c_2} \langle\psi,E_n g_{n+1}E_n\psi\rangle
+  \frac{c_2}{2} \langle\psi, \left(\sum_{k=n_\ell}^{n}E_k\right)^2\psi\rangle.
\eeann
We estimate each of the four terms above as follows. For the first term, we
apply Assumption (i), whereas the second term is immediately combined with the left-hand-side. With the
third term, we use Assumption (iii), and for the fourth term, we again use the mutual
orthogonality of the $E_n$. After some reordering, we have
\be \label{MMOneEn}
(1-\frac{c_1}{2} -   \frac{\epsilon^2}{2c_2}) \Vert E_n\psi\Vert^2 
 - \frac{c_2}{2}\sum_{k=n_\ell}^{n}\Vert E_k\psi\Vert^2 \leq \frac{1}{2c_1 \gamma} \langle\psi, h_{n+1}\psi\rangle.
 \ee
Summing both sides of \eqref{MMOneEn} from $n =0, \ldots, \, N-1$, yields  
\be
\langle \psi ,H_N\psi\rangle \geq 2c_1 \gamma \left[1-\frac{c_1}{2} -   \frac{\epsilon^2}{2c_2}-\frac{c_2(1+\ell)}{2}\right] \Vert \psi\Vert^2. 
\ee
Maximizing this lower bound leads to the choice of
$c_1=1-\epsilon\sqrt{1+\ell}$ and $c_2=\epsilon/\sqrt{1+\ell}$.
This produces the inequality stated in the theorem.
\end{proof}

\section{Discussion: stability, examples, and the spectral flow}\label{sec:stability}

Recently, methods were introduced, see e.g.  \cite{bravyi:2010,michalakis:2013}, which allow for a proof of stability for gapped
ground states of frustration-free quantum spin systems satisfying a Local Topological Quantum Order (LTQO) condition. Here,
the term stability refers to the property that the there is a lower bound of the spectral gap, uniform in $\Lambda$, for finite-volume
Hamiltonians of the form
$$
H_\Lambda(s) = \sum_{X\subset \Lambda} \Phi(X) + s \Psi(X),
$$
where $\Psi$ is any other short-range interaction and $|s|< s_0$ for some $s_0 >0$.
Roughly speaking,
this LTQO condition on $\Phi$ amounts to a precise formulation of the notion that, in this situation, degenerate ground states cannot 
be distinguished by
local operations. In a forth-coming work \cite{nachtergaele:inprep}, a systematic refinement of these techniques will be presented which generalize previous
results. For example, a stability result for gapped ground state phases of quantum spin models with a spontaneously broken discrete symmetry 
is contained in \cite{nachtergaele:inprep}. 

In this section, we give an indication of how analogous results extend to models of lattice fermions.
More details on this extension will be given in \cite{young:inprep}. We begin by defining general frustration-free fermion models. Next, we give a large class of
examples of gapped, frustration-free fermion which may serve as unperturbed models for stability results. Then, we introduce several 
types of symmetry, which may or may not be broken in a given model, and discuss briefly how the LTQO condition should be modified
in this context. Lastly, we discuss the construction of the spectral flow automorphism which is main tool of analysis in proofs of stability. 

In the same way as for quantum spin systems, an even fermion interaction $\Phi$ on $\Gamma$ is defined to 
be {\em frustration-free} if $\Phi$ is finite-range and for $\Lambda \in \cP_0( \Gamma)$ we have
\be
\infspec (H_\Lambda)  =  \sum_{X\subset \Lambda} \infspec (\Phi(X)).
\ee
It is standard to normalize the interactions so that each interaction term is
non-negative, i.e. $\Phi(X) \geq 0$ for each $X \in \cP_0( \Gamma),$ and further
that the ground state energy of $H_{\Lambda}$ vanishes. In this case, a ground
state eigenvector for $H_{\Lambda}$ is necessarily in the kernel of each of its
terms. Since we do not assume the existence of simultaneous eigenvectors
for other energy values, this set-up does not imply that the interaction terms
commute.

We now consider two classes of gapped frustration-free fermion models with finite-range interactions.
As already indicated in Section~\ref{sec:martingale_method}, a first class of examples are 
those one-dimensional fermion systems with finite-range interactions that are mapped to a gapped, frustration-free 
quantum spin chain by the Jordan-Wigner transformation. Note that even interaction terms $\Phi([a,b])\in\cA_{[a,b]}^+$ 
are mapped into $\tilde\Phi([a,b])\in\cA^{\rm spin}_{[a,b]}$. Here {\em evenness} of the interactions is important to preserve its finite-range 
property. It also
implies a local discrete symmetry for the quantum spin model; more on this below. 
The properties of the ground states of such models can be used to verify the conditions of the martingale method discussed in Section \ref{sec:martingale_method}.
This provides a wide class of gapped fermion systems in one-dimension. Another
approach for constructing frustration-free one-dimensional fermion systems
defines fermionic Matrix Product States using graded vector spaces \cite{bultinck:2017}.

A second class of interesting frustration-free fermion models with a spectral
gap is obtained by considering quasi-free systems with two (or more) bands
separated by a gap $\gamma>0$.
For clarity, we will use two sets of fermion creation and annihilation operators
$\{b_k, b^*_k\mid k\in B\}$ and $\{c_l, c^*_l\mid l\in C\}$, labeled by index
sets $B$ and $C$, respectively, which together span $\cA_\Gamma$, where $\Gamma$
is a lattice of arbitrary dimension.
These operators are defined as follows. There are two subsets of
$\ell^2(\Gamma)$, $\{f_k| k\in B\}$, and $\{g_l\mid l\in C\}$ such that: 

(i) $\Vert f_k\Vert = 1, \Vert g_l\Vert =1$, for all $k\in B$ and $l\in C$; 

(ii) $\span (\{f_k \mid k\in B\} \cup \{ g_l\mid l \in C\})$ is dense in
$\ell^2(\Gamma)$; 

(iii) $\langle f_k, g_l\rangle =0$, for all $k\in B$ and $l\in C$; 

(iv) there exists $R\geq 0$, and $x_k, y_l \in\Gamma$, such that $\supp f_k
\subset B_{x_k}(R)$, and $\supp g_l \subset B_{y_l}(R)$, for all $k\in B$ and
$l\in C$. 

 Then, the new creation and annihilation operators are defined by
\bea
b^*_k = \sum_{x\in\Gamma} f_k(x) a^*_x && b_k = \sum_{x\in\Gamma} \overline{f}_k(x) a_x \\
c^*_l = \sum_{x\in\Gamma}g_l(x) a^*_x && c_l = \sum_{x\in\Gamma}\overline{g}_l(x) a_x .
\eea
Next, define an interaction $\Phi$ by setting
\be
\Phi(B_{x_k}(R)) = \idty - b_k^* b_k, \quad \Phi(B_{y_l}(R)) = c_l^* c_l,
\ee
and $\Phi(X) = 0$ if $X$ is not a ball of radius $R$ centered at a site $x_k$ or
$y_l$, for any $k\in B,$ or $l\in C$. In a standard application, the
functions $f_k$ and $g_l$ are the orbitals in the valence and conduction band,
respectively. It is then straightforward to construct a quasi-free state $\omega_\Lambda$ on $\cA_\Lambda$ satisfying 
\be
\omega_\Lambda(b_k^* b_k)=1, \quad  \omega_\Lambda(c_l^* c_l )=0.
\label{fermion_ffgs}\ee 
Clearly, the Hamiltonians
\be
H_\Lambda = \sum_{X\in \subset\Lambda}  \Phi(X)
\ee
are non-negative. Since \eq{fermion_ffgs} implies that $\omega_\Lambda(H_\Lambda)=0$, these models are frustration-free.
The orthogonal complement to the ground state given by $\omega_\Lambda$, is
spanned by the Fock space vectors with at least one hole in the valence band or one particle in the conduction band. In this case, the gap above the ground state energy is 
$\gamma =1 $.

An interesting class of examples of models of this type are the so-called flat-band Hamiltonians studied by Mielke and Tasaki \cite{mielke:1993}.
A thorough study of the conditions under which there exists a spanning set of compactly supported orbitals for a band structure was recently
carried out by Read \cite{read:2017}. He shows that certain band structures in
two or more dimensions, which yield ground states with certain types of
topological order, cannot be spanned by compactly supported orbitals. It is also worth noting that if the goal is to show stability of the gap under small perturbations, one can also deal with some cases in which the bands are spanned by orbitals that are not compactly supported.
If the orbitals are well approximated by compactly supported functions and the gap is not too small, then one can treat the error as part of the perturbation.

We now turn to symmetry. In the classification of gapped ground state phases it is often important to consider symmetries of the model. 
These may or may not be spontaneously broken in the ground state. Common examples of the symmetries we will consider 
include the following:

\begin{enumerate}
\item {\em Parity of the fermion number:} We will always assume that all terms in the interaction are even: $\Phi(X) \in \cA_X^+$. This 
means that all interaction terms supported in $X\in\cP_0(\Gamma)$ commute with $(-1)^{N_X}$.
\item {\em Local symmetries:} A local symmetry is described by a representation
of a finite group $G$ consisting of automorphisms, $\beta_g^x$ for $g\in G$,
acting on $\cA_{\{(x,k)\mid k=1,\ldots, n\}}$, for each $x\in\Gamma$. This invariance can equivalently be described by the commutation of the interaction terms with a unitary representation of $G$: $U_X(g)\Phi(X)=\Phi(X)U_X(g)$. 
The parity of the fermion number is an example of such a discrete local symmetry.
\item {\em Translation invariance:} Often $\Gamma$ is a lattice such as $\cL= \Ir^\nu$ or $\cL= \Ir^\nu/(L \Ir^\nu)$, or $\Gamma$ 
contains a lattice as a factor: $\Gamma = \cL\times \Gamma_1$. In either case,
there is a natural action of $\cL$ on $\Gamma$ which, for simplicity, we will denote by addition. Then, $\cL$ acts on $\cA_\Gamma$ as a group of translation automorphisms $\beta_x, x\in \cL$. Translation 
invariance of the interaction is then expressed by $\beta_x(\Phi(X)) = \Phi(X+x)$.
\item {\em Space inversion and other lattice symmetries:} Besides translations, $\Gamma$ may often possess other discrete symmetries, such as 
inversion ($x\mapsto -x$) or rotation by certain angles.
\item {\em Time reversal invariance:} If there is a basis in Fock space with respect to which the matrix of the Hamiltonian is real, meaning that there is a complex conjugation with which it commutes, we have a symmetry between the forward and backward dynamics, i.e., $\tau_t (\beta(A))
= \beta (\tau_{-t} (A))$, for the corresponding anti-automorphism.
\end{enumerate}

When the goal is to prove stability of the gapped ground 
state phases in cases that allow for a spontaneously broken discrete symmetry,
one is led to assume a slight modification of the 
LTQO condition introduced in \cite{bravyi:2010}.
For example, it no longer makes sense to assume LTQO for arbitrary
local observables; one should restrict attention to those observables that preserve the symmetry. 
Mimicking methods in \cite{nachtergaele:inprep}, one can prove that the spectral gap and the structure of the ground state phases are stable under sufficiently
small perturbations of the interaction, if one additionally assumes that perturbations preserve the symmetry. A precise statement
of this result will appear in \cite{young:inprep}.

Finally, we turn to the main tool used in the recent proofs of stability of gapped phases \cite{bravyi:2010,michalakis:2013}, the so-called {\em spectral flow} 
or {\em quasi-adiabatic continuation} \cite{hastings:2005, bachmann:2012}. To define the spectral flow, consider a one-parameter family
of interactions  $\Phi_s: \cP_0(\Gamma)\to \cA_\Gamma$, such that for each $X\in \cP_0(\Gamma)$, $\Phi_s(X)$ is differentiable
with respect to $s\in [0,1]$. In this case, Hamiltonians of the form $H_\Lambda(s) = \sum_{X\subseteq \Lambda}\Phi_s(X)$ 
are defined on any finite volume $\Lambda \in \cP_0(\Gamma)$.

The situation of interest is where the spectrum of $H_\Lambda(s)$ is composed of two parts, $\Sigma_1(s)$ and $\Sigma_2(s)$, 
separated by a gap bounded below by a constant $\gamma > 0$, for all $s\in [0,1]$.  Let $P(s)$ denote the spectral projection of $H_\Lambda(s)$ corresponding to the set $\Sigma_1(s)$.
Due to the spectral gap assumption, general results imply that $P(s)$ is unitarily equivalent
to $P(0)$ for all $s\in [0,1]$, i.e., there exists a curve of unitaries $U(s)$ satisfying
\be \label{follow}
P(s) = U(s) P(0) U(s)^*.
\ee
The spectral flow is constructed in terms of a particular choice of unitaries satisfying (\ref{follow}). 

As discussed e.g. in \cite{bachmann:2012}, one choice for a family of unitaries satisfying (\ref{follow}) 
is obtained as the solution of the Schr\"odinger equation for unitaries $U(s)$, with $s$ playing the role of time.
The Hamiltonian for this Schr\"odinger equation, $D_\Lambda(s)$, is  defined by an expression of the following form:
\be \label{sf_gen}
D_\Lambda(s) = \int_{-\infty}^{\infty} \tau_t^{H_\Lambda(s)}(H'_\Lambda(s)) \, W_\gamma(t) \, dt,
\ee
where $W_\gamma(t)\in L^1(\bR)$ is a well-chosen function which decays faster than any power law as $t\to \pm\infty$.
In terms of the corresponding unitaries, the spectral flow automorphism is then defined as $\alpha_s^\Lambda(A) = U(s)^*AU(s)$ for all $A\in \cA_\Lambda$.

A crucial technical result, used in all proofs of stability, is the fact that this spectral flow
satisfies a Lieb-Robinson bound with a decay function that is explicit. Since the Hamiltonian terms
corresponding to (\ref{sf_gen}) are not strictly local, such an estimate is not a direct application
of known results. With some effort, one can show that if the family of interactions 
$\Phi_s$ and $\Phi'_s$ decay sufficiently fast (usually expressed in terms of an $F$-function that
decays at an exponential rate), then $D_\Lambda(s)$ can be realized as a local Hamiltonian 
associated to an interaction $\Psi_{\Lambda}$
\be
D_\Lambda(s) = \sum_{X\subseteq \Lambda} \Psi_{\Lambda}(X,s) .
\ee
Moreover, there exists an $F$-function, denoted by $F_\Psi$, for which $\| \Psi_{\Lambda}(s)\|_{F_{\Psi}}< \infty$ uniformly in the finite volume $\Lambda$.
As the explicit $F_{\Psi}$ decays sub-exponentially, one obtains Lieb-Robinson bounds, i.e. locality estimates,
for the spectral flow that decay as fast.  

In the construction of the interaction $\Psi_{\Lambda}(\cdot, s)$, for fermions,
one uses the conditional expectation discussed in Section \ref{sec:cond_exp} in combination with the Lieb-Robinson bounds of Section \ref{sec:LRbounds}. As far as we are aware, the first use of a conditional expectation to construct an interaction 
$\Phi$ from a set of local Hamiltonians of a quantum system appeared in \cite{araki:1978}.
To decompose $D_\Lambda(s)$ into strictly local terms, it is crucial that the initial interaction terms $\Phi(X,s)$ be even. 
In this case, $\Psi_\Lambda(\cdot, s)$ is also even. Of course, for models with other symmetries, one must check that 
this localizing operation preserves the symmetry.  Given that this is the case, the proof of stability then proceeds in the 
same way as in the case of quantum spin systems. Explicit estimates and
further applications will be given in \cite{young:inprep}.

\section*{Acknowledgements}
We thank Detlev Buchholz for pointing out to us the use of a conditional expectation in the work of Araki and Moriya \cite{araki:2003}.


%
\providecommand{\bysame}{\leavevmode\hbox to3em{\hrulefill}\thinspace}
\providecommand{\MR}{\relax\ifhmode\unskip\space\fi MR }
\providecommand{\MRhref}[2]{%
  \href{http://www.ams.org/mathscinet-getitem?mr=#1}{#2}
}
\providecommand{\href}[2]{#2}

\end{document}